%% file: main.tex
\documentclass{ifacconf}

\usepackage{graphicx} 
\usepackage{natbib} 

\usepackage{amsmath, amsfonts}
\usepackage{algorithmic}
\usepackage{algorithm}
\usepackage{array}
\usepackage{caption}
\usepackage{subcaption}
\usepackage{textcomp}
\usepackage{stfloats}
\usepackage{url}
\usepackage{verbatim}
\usepackage{bm}
\usepackage{subfiles}
\usepackage{comment}
\usepackage{enumitem}

\newtheorem{assumption}{Assumption}

\newtheorem{lemma}{Lemma}
\newtheorem{remark}{Remark}
\newtheorem{theorem}{Theorem}
\begin{document}
\begin{frontmatter}
    \title{\mbox{Passivity-based Semi-autonomous Rotational}\\ 
    \mbox{Motion Navigation for Rigid-body \nolinebreak Networks:}\\
    Stability and Human Passivity Analysis\thanksref{footnoteinfo}} 

    \thanks[footnoteinfo]{This work has been supported by JSPS KAKENHI Grant Number 24K00906. Approval of all ethical and experimental procedures and
protocols was granted by the Human Subjects Research
Ethics Review Committee in Institute of Science Tokyo
under Application No. 2025151, and performed in line
with the Helsinki on Ethical Principles for Medical Research 
and Ethical Guidelines for Medical and Biological
Research Involving Human Subjects.}

    \author[First]{Reiji Terunuma}
    \author[First]{Yuta Nakamura}
    \author[First]{Takeshi Hatanaka}

    \address[First]{Institute of Science Tokyo (e-mail: \{terunuma,~nakamura.y\}@hfg.sc.e.titech.ac.jp,
        hatanaka@sc.e.titech.ac.jp).}

    \subfile{text/abstract}

\end{frontmatter}

\subfile{text/introduction}
\subfile{text/scenario}
\subfile{text/passivity-based_control}
\subfile{text/simulation_design}
\subfile{text/conclusion}


\bibliography{ifacconf} 

\end{document}

%% file: text/abstract.tex
\begin{abstract}                
    This paper presents a novel passivity-based semi-autonomous attitude control framework,
    with a particular focus on attitude kinematics defined on the special orthogonal group $SO(3)$.
    While human-robot interaction facilitates the successful execution of complex tasks, 
    ensuring stability of human-in-the-loop systems on the $SO(3)$ manifold 
    remains a largely unsolved challenge.
    We first propose a new control architecture in which a multi-robot system preserves invariance of the average information fed back to the human operator through so-called stealthy control, and the human intervention is mediated through a virtual leader, which is coupled with the robots via a passivity-based attitude synchronization law.
    We then rigorously prove closed-loop stability of the proposed human-in-the-loop system under the assumption that the human behaves as a passive system.
    To support this analysis, simulation studies 
    are conducted to identify the human operator as a dynamical system, and to examine passivity properties of the identified model.
\end{abstract}

%% file: text/introduction.tex
\section{Introduction}
\label{chp:introduction}
Cooperative control of multi-robot systems has reached a mature stage, leading to the development of numerous control strategies that can accomplish increasingly complex real-world tasks,
such as multi-robot patrolling (\cite{Marino:2012}),
cooperative aerial transportation (\cite{Sun:2025}),
and coordinated multi-drone image sampling (\cite{Hanif:2025}).
Despite their effectiveness in static and well-structured environments, these systems often encounter challenges in dynamic and/or unstructured scenarios.
In such situations, human operators play a crucial role, as their flexible reasoning, adaptability,  and situational awareness enable the successful execution of complex tasks under uncertainty (\cite{Daihya:2023}).
Various system architectures for semi-autonomous navigation of multiple robots have been reported 
in the literature (\cite{Lee:2005}, \cite{Franchi:2012b} and \cite{Rodriguez:2010}). 
These systems are typically designed according to a fundamental principle that high-level navigation is delegated to a human operator, while low-level motion coordination is handled by autonomous control. 
In such human-in-the-loop systems, conflicts between manual control and autonomous control can arise, potentially compromising the closed-loop stability.
\cite{Antonelli:2008,Music:2017} proposed a so-called null-space-based behavioral control that decouples conflicts among robotic subtasks.
Building on this approach, \cite{Terunuma:2025} presented stealthy coverage control for rigid-body networks, which decouples 
manual navigation from autonomous coverage control. 
Specifically, the autonomous control is constrained to the null space of the average state of the bodies, which is fed back to the human operator.
While this approach, which decouples autonomous control from stability of human-in-the-loop systems, is expected to drastically simplify the stability analysis, ensuring closed-loop stability for full three-dimensional rigid-body motion remains a critical challenge.

Passivity-based approaches have been widely employed to guarantee closed-loop stability for human-in-the-loop robotic systems.
In the context of bilateral teleoperation, \cite{Anderson:1989} ensured stability by assuming the human passivity. 
However, validating this passivity assumption is nontrivial, as human operators do not always satisfy passivity (\cite{Hatanaka:2017}).
To address this issue, control frameworks that account for the possible human passivity shortage have been proposed 
(\cite{Hatanaka:2024}).
Nevertheless, these studies are primarily limited to translational motion or two-dimensional attitude kinematics (\cite{kanazawa:2023}). 
As a result, a framework that ensures stability under human intervention while robots perform full three-dimensional autonomous collaborative maneuvers, including rotational motion, has not yet been established.

In this paper, we focus on semi-autonomous attitude control of robots interacting with a human operator, as a preliminary step toward addressing full three-dimensional rigid-body motion.
To this end, we consider attitude kinematics on the special orthogonal group $SO(3)$
following the foumulation in (\cite{Hatanaka:2015}).
A novel human-in-the-loop control architecture is then designed based on the concept of stealthy control and the inherent passivity of the rigid-body rotational motion. 
We then rigorously prove closed-loop stability under the assumption of human passivity.
Furthermore, human-in-the-loop simulations are conducted to identify the human operator dynamics from operation data and to examine the validity of the human passivity assumption.


The contributions of this paper are summarized as follows.
\begin{enumerate}[label=\roman*)]
    \item A novel passivity-based semi-autonomous attitude control framework for 
    human-in-the-loop rigid-body systems is proposed.
    \item Closed-loop stability of the proposed human-in-the-loop control architecture is rigorously established.
    \item  Human passivity in attitude control is investigated for the first time through human-in-the-loop simulations.
\end{enumerate}


%% file: text/scenario.tex
\section{Scenario and Objective}
\label{chp:model}
\begin{figure}
    \centering
    \includegraphics[width=0.845\linewidth]{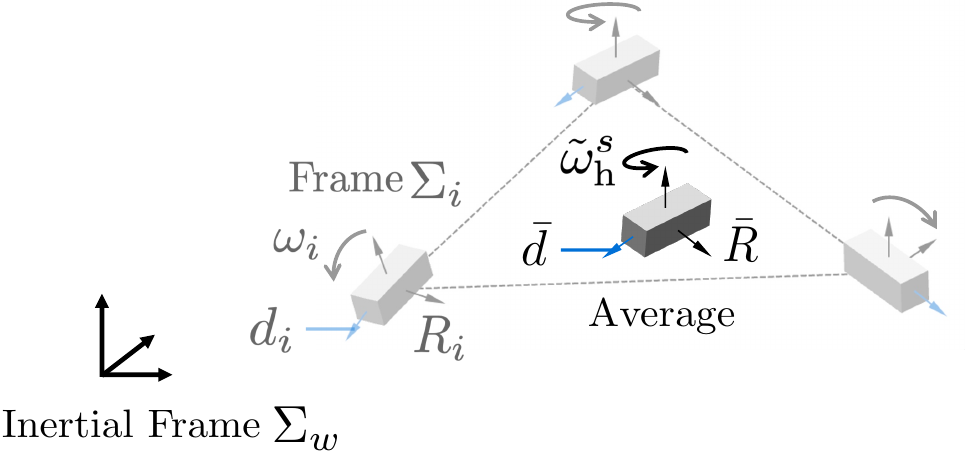}
    \caption{Average and rigid bodies in 3D space.
    In this paper, we focus exclusively on the rotational motion.}
    \label{fig:rigid_body_general}
\end{figure}

We consider a system of $n$ rigid bodies, 
as illustrated in Fig.~\ref{fig:rigid_body_general}.
Let $\Sigma_w$ be the inertial coordinate frame and $\Sigma_i$ be 
the body-fixed frames of each body $i \in \mathcal{N}=\{1,2, \ldots, n\}$.
The orientation of $\Sigma_i$ relative to $\Sigma_w$ is represented by the rotation matrix $R_i \in SO(3)$.
The rotational motion of body $i$ is then known to be formulated as
\begin{align}
    \dot{R}_{i} & =R_{i} \hat{\omega}_{i},
    \label{eq:each_robot_dynamics}
\end{align}
where $\omega_i\in\mathbb{R}^3$ is the angular body velocity and is regarded as the control input.
The symbol $\wedge$ (wedge) is the operator from $\mathbb{R}^3$ to the space of $3\times 3$ skew-symmetric matrices.
The symbol $\vee$ (vee) is the inverse operator to the wedge.

We next define a fixed vector in $\Sigma_i$.
Without loss of generality, the vector is identified with the $z$-axis of the body frame $\Sigma_i$. The vector is then represented in $\Sigma_w$ by $d_i=R_i\mathrm{e}_3\in\mathbb{R}^3$ with $\mathrm{e}_{3} = [0\ 0\ 1]^{\top}$.
The dynamics of the collective vector $d=[d_1^\top\ d_2^\top\ \cdots\ d_n^\top]^\top\in \mathbb{R}^{3n}$ can be expressed as
\begin{align*}
    \dot{d} &= S(d)\Omega,
\end{align*}
where 
\begin{align*}
    S(d) & = \begin{bmatrix}
        -R_1\hat{\mathrm{e}}_3 & O_{3} & \cdots & O_{3}\\
        O_{3} & -R_2\hat{\mathrm{e}}_3 & \cdots & O_{3}\\
        \vdots & \vdots & \ddots & \vdots\\
        O_{3} & O_{3} & \cdots & -R_n\hat{\mathrm{e}}_3
    \end{bmatrix}\in\mathbb{R}^{3n\times 3n},\\
    \Omega & = \begin{bmatrix}\omega_{1}^{\top}&\omega_{2}^{\top}&\ldots&\omega_{n}^{\top}\end{bmatrix}^{\top}\in\mathbb{R}^{3n},
\end{align*}
and $O_3$ is the $3$-by-$3$ zero matrix.
We also define the average of the vectors $d_1, d_2,\dots, d_n$ as
\begin{align*}
\bar d = \dfrac{\sum_{i\in\mathcal{N}} d_i}{\|\sum_{i\in\mathcal{N}} d_i\|}.
\end{align*}

In this paper, we suppose that the rigid bodies and a human operator cooperatively determine the velocity input $\omega_i$.
To this end, the input $\omega_i$ is defined as $\omega_i = \omega_{\mathrm{a}i} + \omega_{\mathrm{h}i}$, where 
$\omega_{\mathrm{a}i}$ denotes the autonomous control for body $i$ and $\omega_{\mathrm{h}i}$ is a signal reflecting the human command. 
The operator is also assumed to primarily manipulate the average vector $\bar d$. 
To this end, the operator generates a spacial velocity command $\omega^s_\mathrm{h} \in \mathbb{R}^3$ by using a command interface.
The interface modifies the raw command $\omega^s_\mathrm{h}$ to $\tilde \omega^s_\mathrm{h}$, and 
the signal $\omega_{\mathrm{h}i}$ is then determined as
\begin{align}
    \omega_{\mathrm{h}i} = R_i^\top \tilde \omega^s_\mathrm{h}.
\end{align}

Let us now design the stealthy control presented in (\cite{Terunuma:2025}).
Define
the collective velocity inputs $\Omega\in\mathbb{R}^{3n}$ as: 
\begin{align*}
    \Omega & = \begin{bmatrix}\omega_{1}^{\top}&\omega_{2}^{\top}&\ldots&\omega_{n}^{\top}\end{bmatrix}^{\top} = \Omega_{\mathrm{h}} + \Omega_\mathrm{a},\\
    \Omega_{\mathrm{h}} & = \begin{bmatrix}\omega_{\mathrm{h}1}^{\top}&\omega_{\mathrm{h}2}^{\top}&\cdots&\omega_{\mathrm{h}n}^{\top}\end{bmatrix}^{\top},\\
    \Omega_{\mathrm{a}} & = \begin{bmatrix}\omega_{\mathrm{a}1}^{\top}&\omega_{\mathrm{a}2}^{\top}&\cdots&\omega_{\mathrm{a}n}^{\top}\end{bmatrix}^{\top}.
\end{align*}
The stealthy control is then designed as
\begin{align*}
    \Omega_\mathrm{a} = A(d)\tilde \Omega_\mathrm{a},
\end{align*}
where $\tilde \Omega_\mathrm{a}$ is an arbitrarily designed control signal depending on the mission.
The matrix 
$A(d)\in\mathbb{R}^{3n\times 3n}$ is defined as
\begin{align*}
    A(d) &= I_{3n} - W(d)(W^\top(d) W(d))^{-1}W^\top(d),
\end{align*}
where 
\begin{align*}
    W(d)&=\left(\dfrac{\partial \bar d}{\partial d}S(d)\right)^\top\in\mathbb{R}^{3n\times3},
\end{align*}
and $I_{3n}$ is the identity matrix of size $3n$.
Then, the dynamics of the average vector $\bar d$ can be expressed as
\begin{align}
    \dot{\bar{d}} &= \dfrac{\partial \bar d}{\partial d}\dot{d} = \dfrac{\partial \bar d}{\partial d}S(d)(\Omega_{\mathrm{h}} + A(d)\tilde\Omega_\mathrm{a}) = \dfrac{\partial \bar d}{\partial d}S(d)\Omega_{\mathrm{h}}.
    \label{eqn:stealth_control}
\end{align}
This implies that the autonomous control does not affect the average vector $\bar d$, meaning that the operator cannot recognize the influence induced by the automatic control, as long as the operator is handling only a quantity associated with $\bar d$. 


The control goal addressed in this paper is to design a semi-autonomous system 
that drives $\bar d$ toward a reference vector $d_r$ desirable for the operator, namely
\begin{align}
    \label{eq:control_objective}
    \lim_{t \to \infty} \| \bar d - d_r\| = 0.
\end{align}

\begin{remark}
Suppose that the visual coverage control is implemented as $\tilde \Omega_\mathrm{a}$ (\cite{Terunuma:2025}), where the vector $d_i$ is identified with the optical axis of an onboard camera. Under this setting, coordinated data sampling is achieved by prioritizing coverage of the area specified by the human operator through the vector $\bar d$.
Meanwhile, by applying an attitude synchronization law (\cite{Hatanaka:2015}) and identifying $d_i$ as the body heading direction, the group can be steered toward directions that are desirable to the human operator.
\end{remark}



%% file: text/passivity-based_control.tex
\section{Passivity-based Semi-Autonomous Attitude Control Architecture}
\label{sec:passivity-based_control}
\begin{figure}[t]
    \centering
    \includegraphics[width=0.98\linewidth]{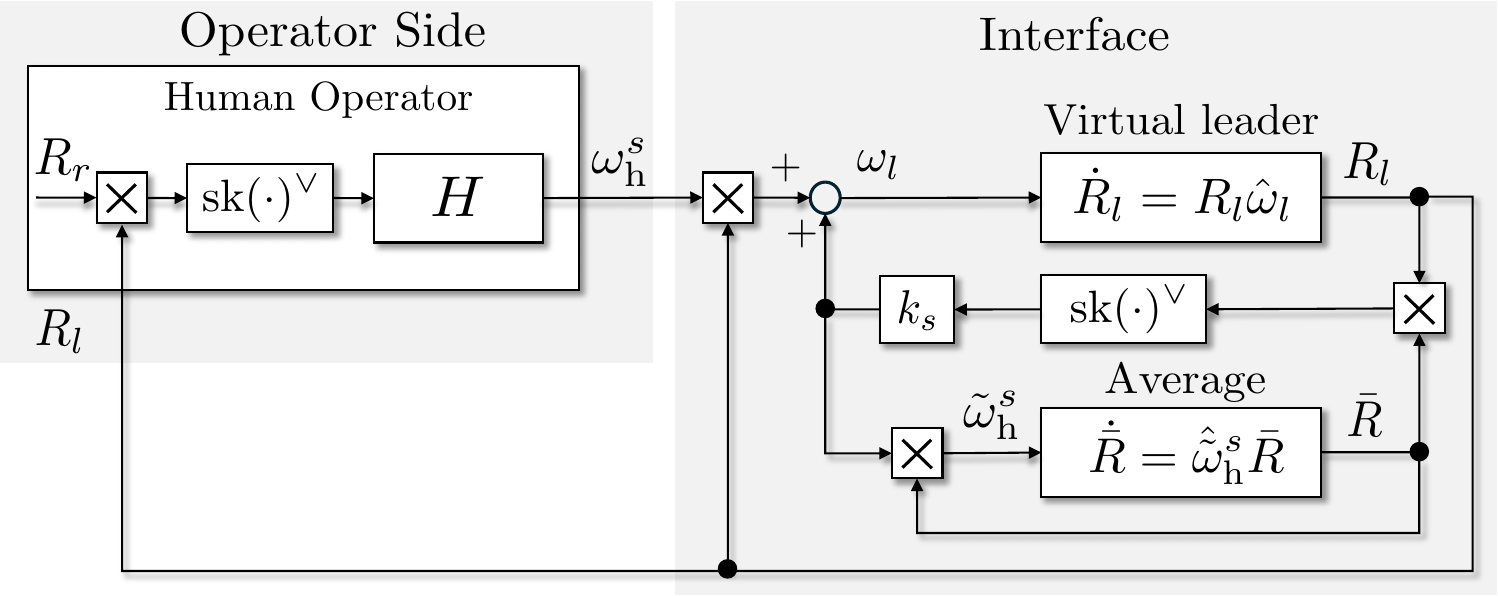}
    \caption{Block diagram of the proposed passivity-based semi-autonomous attitude control architecture.
    The block $H$ represents a mapping from the attitude deviation to the human command $\omega_\mathrm{h}^s$.}
    \label{fig:block_diagram}
\end{figure}
This section proposes a passivity-based semi-autonomous attitude control architecture that ensures stability of the 
human-in-the-loop system under human passivity.

For a technical reason, we begin by introducing quasi-average rotation matrix
$\bar R \in SO(3)$. 
Let the initial $z$-axis $\bar R\mathrm{e}_3$ be aligned with $\bar d$ at the initial time, and
the matrix $\bar R$ be the solution to the differential equation
\begin{align}
    \dot{\bar{R}} = \hat{\tilde{\omega}}^s_\mathrm{h} \bar R.
    \label{eq:actual_body_dynamics_world}
\end{align}
Then, because of (\ref{eqn:stealth_control}), the $z$-axis $\bar R\mathrm{e}_3$ remains to be the same as $\bar d$. The remaining degree of freedom corresponding to the rotation around the $z$-axis may differ from the actual average of the rotation matrices $R_i\ (i \in \mathcal{N})$. However, it does not matter in the present control scenario since the goal is to align the vector $\bar d$ only.
We also define a quasi-reference rotation matrix $R_r \in SO(3)$ such that its $z$-axis is the same as the human reference $d_r$ and the rotation around the $z$-axis is arbitrarily chosen.


We next introduce a virtual leader 
whose attitude dynamics 
are given by
\begin{align}                    
    \dot{R}_{l} & =R_{l}\hat{\omega}_{l}, 
    \label{eq:rigid_body_dynamics_world}
\end{align}
where $R_l \in SO(3)$ and $\omega_l$ denote the attitude and angular body velocity of the virtual leader. We also define $d_l = R_l\mathrm{e}_3$.
The introduction of the leader is inspired by (\cite{Hatanaka:2024}), and this will play an important role in proving closed-loop stability.

The human operator is assumed to perceive the vector $d_l$, and commands the angular velocity $\omega^s_\mathrm{h}$
so that $d_l$ synchronizes the reference $d_r$.
In order to fill the gap between $d_l$ and $\bar d$ in (\ref{eq:control_objective}), 
we design an attitude synchronization law as
\begin{subequations}
\begin{align}
    \tilde \omega_{\mathrm{h}}^s   &= k_s \bar R\ \text{sk}(\bar R^\top R_{l})^{\vee}, \label{eqn:tilde_omega_h}\\
    \omega_l &= k_s\,\text{sk}(R_{l}^\top \bar R)^{\vee}+R_l^\top{\omega}^s_\mathrm{h},
\end{align}
\label{eq:rigid_body_dynamics_control_world}
\end{subequations}
where $\text{sk}(A)$ is the skew-symmetric part of the square matrix $A$, i.e., $\text{sk}(A) = (A - A^\top)/2$.
The gain $k_s>0$ is associated with the coupling strength between the average and virtual leader.
Note that the rotational motion (\ref{eq:actual_body_dynamics_world}) with (\ref{eqn:tilde_omega_h}) is known to be formulated as
\begin{align}
    \dot{\bar{R}} = k_s \bar R\,\text{sk}(\bar R^\top R_{l}).
    \label{eqn:hatanaka_edit}
\end{align}
We see from this and (\cite{Hatanaka:2015}) that,
if $\omega_\mathrm{h}^s\equiv 0_3$, the control law (\ref{eq:rigid_body_dynamics_control_world}) achieves attitude synchronization between $\bar R$ and $R_l$.
The proposed control architecture is illustrated in Fig.~\ref{fig:block_diagram}.

Let us now introduce the following assumption, where $R_{rl}=R_r^\top R_{l}$
denotes the relative rotation and
$\omega_\mathrm{h}^b=R_l^\top\omega_\mathrm{h}^s$ is the body velocity command.
\begin{assumption}
    \mbox{}
    \begin{itemize}
        \item $R_{r}$ is constant.
        \item $R_{rl} + R_{rl}^{\top}$ is positive definite at the initial time and the human operator manipulates the virtual leader to ensure this condition at all subsequent times.
        \item The operator identically sends $\omega_{\mathrm{h}}^s=0_3$ only if $d_{l}=d_r$ identically holds.
        \item If the input to the human $\operatorname{sk}(R_{l}^{\top}(t) R_{r})^{\vee}$ is bounded, the output $\omega_{\mathrm{h}}^s$ and the internal state of the human are bounded.
        \item The human operator is time invariant and passive from $\operatorname{sk}(R_{l}^{\top}(t) R_{r})^{\vee}$ to $\omega_{\mathrm{h}}^b$, i.e. $\exists \beta \geq 0$ such that
            \begin{equation}
                \int_{0}^{\tau} \{\operatorname{sk}(R_{l}^{\top}(t) R_{r})^{\vee}\}^{\top}
                \omega_{\mathrm{h}}^b(t) \, dt \geq -\beta,\quad \forall \tau \geq 0.
                \label{eq:passivity_human}
            \end{equation}
    \end{itemize}
    \label{assum:passivity_human}
\end{assumption}
We will examine the validity of the fifth item of Assumption \ref{assum:passivity_human} in the next section.

\begin{lemma}
\label{lem:1}
    Consider the semi-autonomous control system consisting of (\ref{eq:actual_body_dynamics_world}), (\ref{eq:rigid_body_dynamics_control_world}) and
    a human operator satisfying Assumption \ref{assum:passivity_human}.
    Then, the matrix $\bar R_{r} + \bar R_{r}^{\top}$ with $\bar R_{r}=R_r^\top \bar R$ remains positive semi-definite,
    provided that the matrix is positive semi-definite at the initial time.
\end{lemma}
{\it Proof.}
    Utilizing Rodrigues' rotation formula, the relative rotation $\bar R_{r}$ is expanded as
    \begin{align*}
        \bar R_{r}=\exp(\xi\theta)=I_3+\hat\xi\sin\theta+\hat\xi^2(1-\cos\theta), 
    \end{align*}
    where $\xi \in \mathbb{R}^3$ is the unit rotation axis and $\theta \in \mathbb{R}$ is the rotation angle.
    For $\bar R_{r}^\top + \bar R_{r}$ to be positive semi-definite, all of its eigenvalues $\{2,2\cos\theta,2\cos\theta\}$ must be nonnegative,
    namely $\cos\theta \geq 0$.
    This condition can be equivalently expressed as
    \begin{align}
        \text{tr}(\bar R_{r})=1+2\cos\theta \geq 1.
        \label{eq:tr_Rf}
    \end{align}
    Here, we introduce a scalar function
    \begin{align*}
        h(\bar R_{r})=\text{tr}(\bar R_{r})-1,
    \end{align*}
    such that the condition (\ref{eq:tr_Rf}) is equivalent to $h(\bar R_{r}) \geq 0$.
    According to Nagumo's theorem (\cite{Nagumo:1942}), the set defined by $h(\bar R_{r}) \geq 0$ is forward invariant 
    if $\dot{h}(\bar R_{r}) \geq 0$ and $\partial h(\bar R_{r})/\partial \bar R_{r}\neq O_3$  hold on the boundary where $h(\bar R_{r}) = 0$. Using (\ref{eqn:hatanaka_edit}),
    the time derivative of $h(\bar R_{r})$ is given by
    \begin{align}
        \dot{h}(\bar R_{r})&=\text{tr}(\dot{\bar{R}}_{r})=k_s\,\text{tr}\{ R_{r}^\top\bar R\,\mathrm{sk}(\bar R^\top R_l)\}\notag\\
        &=k_s\,\text{tr}\{\bar R_{r}\,\text{sk}(\bar R_{r}^\top R_{rl})\}
        = \dfrac{1}{2}k_s\,\text{tr}(R_{rl}-R_{rl}^\top \bar R_{r}^2).
        \label{eq:dot_h}
    \end{align}
    On the boundary, the rotation angle satisfies the condition $\cos\theta = 0$. 
    Under this condition, $\bar R_{r}^2$ in (\ref{eq:dot_h}) is simplified as
    \begin{align*}
        \bar R_{r}^2&=I_3+\hat\xi\sin(2\theta)+\hat\xi^2\{1-\cos(2\theta)\}\\
        &=I_3+2\hat\xi^2
        =2\xi\xi^\top-I_3.
    \end{align*}
    Noting that $2\xi\xi^\top\!-I_3$ is a symmetric matrix, (\ref{eq:dot_h}) can be rewritten as
    \begin{align*}
        \dot{h}(\bar R_{r})&=\dfrac{1}{4}k_s\,\text{tr}(R_{rl}+R_{rl}^\top-R_{rl}\bar R_{r}^2-R_{rl}^\top \bar R_{r}^2)\\
        &=\dfrac{1}{2}k_s\,\text{tr}\{\text{sym}(R_{rl})(I_3-\bar R_{r}^2)\}\\
        &= k_s\,\text{tr}\{\text{sym}(R_{rl})(I_3-\xi\xi^\top)\}\\
        &\geq k_s\,\lambda_{\min}(\text{sym}(R_{rl})) \text{tr}(I_3 - \xi\xi^\top),
    \end{align*}
    where $\lambda_{\min}(A)$ denotes the minimum eigenvalue of a square matrix $A$, 
    and $\text{sym}(A)$ is the symmetric part of the square matrix $A$, i.e., $\text{sym}(A) = (A + A^\top)/2$.
    The inequality was derived using the property that $\text{tr}(AB) \geq \lambda_{\min}(A)\text{tr}(B)$,
    where $A\in\mathbb{R}^{n\times n}$ is a real symmetric matrix and $B\in\mathbb{R}^{n\times n}$ is a square matrix (\cite{Hatanaka:2015}).
    Under the second item of Assumption \ref{assum:passivity_human}, $\text{sym}(R_{rl})$ is positive definite.
    Given that $I_3-\xi\xi^\top$ is positive semi-definite,
    their product satisfies $\text{tr}\{\text{sym}(R_{rl})(I_3-\xi\xi^\top)\}\geq0$,
    which means that $\dot{h}(\bar R_{r}) \geq 0$ holds on the boundary where $h(\bar R_{r}) = 0$. 
    Moreover, the gradient of $h(\bar R_{r})$ satisfies $\partial h(\bar R_{r})/\partial \bar R_{r} = I_3$ globally on $SO(3)$.
    By invoking Nagumo's theorem, we can show that $R_{r}^{\top}\bar R + (R_{r}^{\top} \bar R)^{\top}$ remains positive semi-definite.
\hfill $\square$

Under Assumption \ref{assum:passivity_human}, we have the following theorem. 
\begin{theorem}
    Consider the semi-autonomous control system consisting of (\ref{eq:actual_body_dynamics_world}),(\ref{eq:rigid_body_dynamics_control_world}) and
    a human operator satisfying Assumption \ref{assum:passivity_human}.
    Then, if $\bar R_{r} + \bar R_{r}^\top$ is positive semi-definite at the initial time,
    the system trajectories achieve (\ref{eq:control_objective}).
\end{theorem}
{\it Proof.}
    Let us define the energy function for a rotation matrix $R \in SO(3)$ as
    \begin{align*}
        \phi(R) := \frac{1}{2} \text{tr}(I_3-R) = \frac{1}{4} \|I_3-R\|_F^2,
    \end{align*}
    where $\|\cdot\|_F$ denotes the matrix Frobenius norm and $\phi(R)$ is non-negative. 
    Using the energy function, we define $S_{r} := \phi(\bar R_{r})$ and $S_{rl} := \phi(R_{rl})$.
    The time derivative of the function $S_{r}$ is given by
    \begin{align*}
        \dot{S}_{r}&=-\frac{1}{2}\text{tr}(\dot{\bar R}_{r})
        =-\frac{1}{2}k_s \, \text{tr}\{\bar R_{r} \, \text{sk}(\bar R_{r}^\top R_{rl})\} \\
        &=-\frac{1}{4}k_s \, \text{tr}\{\bar R_{r}(\bar R_{r}^\top R_{rl}-R_{rl}^\top \bar R_{r})\} \\
        &=-\frac{1}{4}k_s \, \text{tr}(R_{rl}-\bar R_{r}R_{rl}^\top \bar R_{r})  \\
        &=-\frac{1}{4}k_s \, \text{tr}\{(\bar R_{r}+\bar R_{r}^\top )(I_3-\bar R_{r}^\top R_{rl})\}\\
        &\hspace{2cm}+\frac{1}{2} k_s\, \text{tr}(\bar R_{r})-\frac{1}{2}k_s\, \text{tr}(R_{rl})  \\
        &=-\frac{1}{4}k_s \, \text{tr}\{(\bar R_{r}+\bar R_{r}^\top )(I_3-\bar R^\top R_{l})\}\\
        &\hspace{2cm}-\frac{1}{2}k_s\, \text{tr}(I_3-\bar R_{r})+\frac{1}{2}k_s \,\text{tr}(I_3-R_{rl})  \\
        &\leq -\frac{1}{2}k_s\,\lambda_{\text{min}}(\bar R_{r}+\bar R_{r}^\top )\phi(\bar R^\top R_{l})\\
        &\hspace{2cm}-k_s\, \phi(\bar R_{r})+k_s \, \phi(R_{rl}).
    \end{align*}
    Similarly, the time derivative of $S_{rl}$ is given by
    \begin{align*}
        \dot{S}_{rl}&=-\frac{1}{2}\text{tr}(\dot R_{rl})=-\frac{1}{2}\text{tr}\{R_{rl} (k_s \, \text{sk}(R_{rl}^\top \bar R_{r}) + \hat{\omega}_\mathrm{h}^b)\}\\
                &= -\frac{1}{2}k_s\, \text{tr}\{R_{rl}\,\text{sk}(R_{rl}^\top  \bar R_{r})\}-\frac{1}{2}\text{tr}(R_{rl}\hat{\omega}_\mathrm{h}^b)\\
                &\leq  -\frac{1}{2}k_s \, \lambda_{\text{min}}(R_{rl}+R_{rl}^\top )\phi(R_{l}^\top \bar R)+k_s \, \phi(\bar R_{r})\\
                &\hspace{2cm}-k_s \, \phi(R_{rl})+\{\text{sk}(R_{rl})^{\vee}\}^\top \omega_\mathrm{h}^b.
    \end{align*}
    Now, (\ref{eq:passivity_human}) is rewritten as 
    \begin{align*}
        S_{\mathrm{h}}=
        -\int_{0}^{\tau} \{\text{sk}(R_{rl}(t))^{\vee}\}^{\top}
                \omega_{\mathrm{h}}^b(t) \, dt+\beta\geq0,
    \end{align*}
    where $S_{\mathrm{h}}$ is the energy function of the human operator.
    Defining the total energy function of the system as $V = S_{r} + S_{rl} + S_{\mathrm{h}} \ge 0$, we have
    \begin{align}
\dot{V}&=\dot{S}_{r}+\dot{S}_{rl}+\dot{S}_\mathrm{h}\notag \\
        &\leq -\frac{1}{2}k_s \, \{\lambda_{\text{min}}(\bar R_{r}+\bar R_{r}^\top )+\lambda_{\text{min}}(R_{rl}+R_{rl}^\top )\}\phi(R_{l}^\top \bar R).
        \label{eq:dotV}
    \end{align}
    Under the second item of Assumption \ref{assum:passivity_human}, 
    $\lambda_{\text{min}}(R_{rl}+R_{rl}^\top )>0$
    and $\lambda_{\text{min}}(\bar R_{r}+\bar R_{r}^\top )\geq 0$ holds from Lemma \ref{lem:1}, which
    means $\dot{V} \le 0$.
    We see from (\ref{eq:dotV}) that $\dot{V} = 0$ is fulfilled if and only if $R_l=\bar R$ holds.

    Let us consider the set of states identically satisfying $\dot{V} \equiv 0$, which implies $\textrm{sk}(R_l^\top \bar R)^\vee =0$.
    In this set, the time derivative of $R_l^\top \bar R$ is given by
    \begin{align*}
        \frac{d}{dt}(R_l^\top \bar R)&=\dot{R}_l^\top \bar R + R_l^\top \dot{\bar R}
        =-R_l^\top\hat{\omega}_\mathrm{h}^sR_l
        =O_3,
    \end{align*}
    which implies that $\omega^s_\mathrm{h} = 0$ must identically holds in the set.
    Utilizing the third item of Assumption \ref{assum:passivity_human}, 
    $d_r = d_l = \bar d$ must identically hold in order that the state keeps meeting $\dot V \leq 0$.
    Then, by invoking the LaSalle's principle, we can show that
    the system trajectories achieve (\ref{eq:control_objective}).
\hfill $\square$

%% file: text/simulation_design.tex
\section{Passivity of the Human Decision Process}
\label{chp:simulation_design}
\begin{figure}[t]
  \centering
  \includegraphics[width=\linewidth]{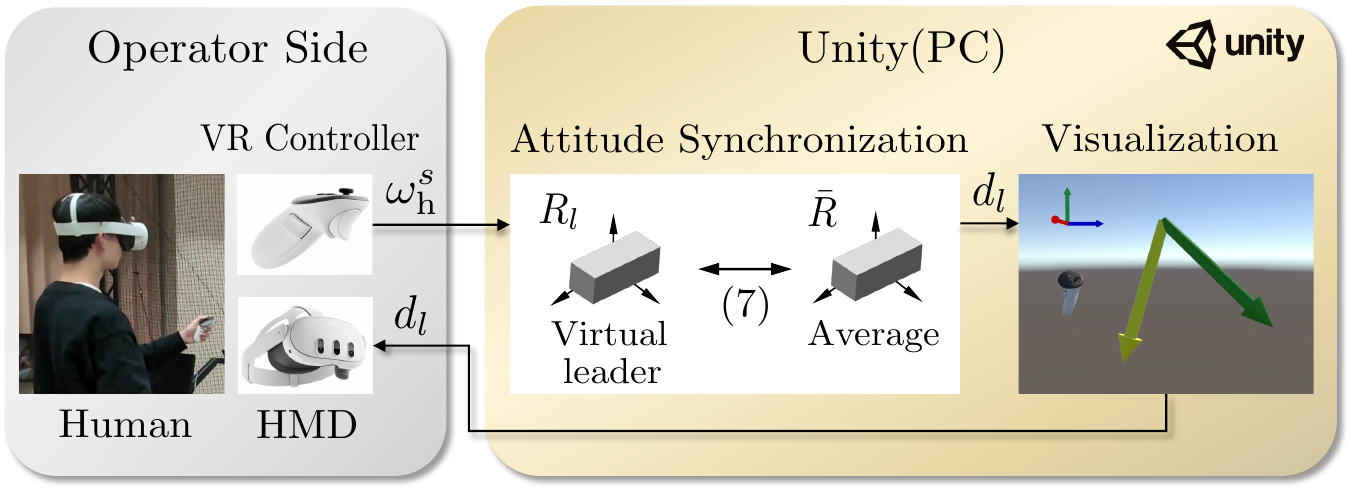}
  \caption{Schematic of the human-in-the-loop simulation system.
      The left block shows the human operator, equipped with the VR controller and head-mounted display (HMD).
      The right block shows the virtual 3D scene built on Unity, 
      where the rigid-body motions are updated in real-time.}
      \label{fig:simulation_system}
\end{figure}
In this section, we examine human passivity assumed in Assumption \ref{assum:passivity_human}.
Specifically, we build human operator models using operation data obtained from the human-in-the-loop simulation. We then analyze passivity of the resulting human models.

\subsection{Simulation Setup}
Fig.~\ref{fig:simulation_system} illustrates a schematic of the present human-in-the-loop simulator, which
integrates a virtual reality (VR) interface, a human operator, and the present controller.
The attitude dynamics 
are simulated on Unity.
A Meta Quest 3S head-mounted display (Meta Platforms, Inc.) is used to provide visual feedback to the human operator, while Touch Plus controller (Meta Platforms, Inc.) serves as the command  interface. 
The simulator executes a loop at 120 Hz, updating both the human command and the rigid-body motions in Unity. 
The initial attitudes of the bodies are set to $\bar R=R_l=I_3$.

The operator commands the spatial angular velocity
$\omega_{\mathrm{h}}^{s}$ by manipulating the
attitude of the controller. The attitude of the controller at the moment of pressing
the trigger button is recorded as a temporal initial orientation $R_{0}\in SO(3)$. 
The spatial angular velocity $\omega_{\mathrm{h}}^{s}$ is generated based
on the attitude of the VR controller $R_{t}\in SO(3)$ as follows:
\begin{align*}
    \omega_{\mathrm{h}}^{s} = k_{\omega}\,\text{log}(R_{t}R_{0}^{\top})^{\vee}, 
\end{align*}
where $k_{\omega} > 0$ is a control gain, which was set to $k_{\omega} =0.8$ in this study.
The norm $\|\omega_{\mathrm{h}}^{s}\|$ was saturated to $1.0$ rad/s.
\begin{figure}
  \centering
  \includegraphics[width=0.65\linewidth]{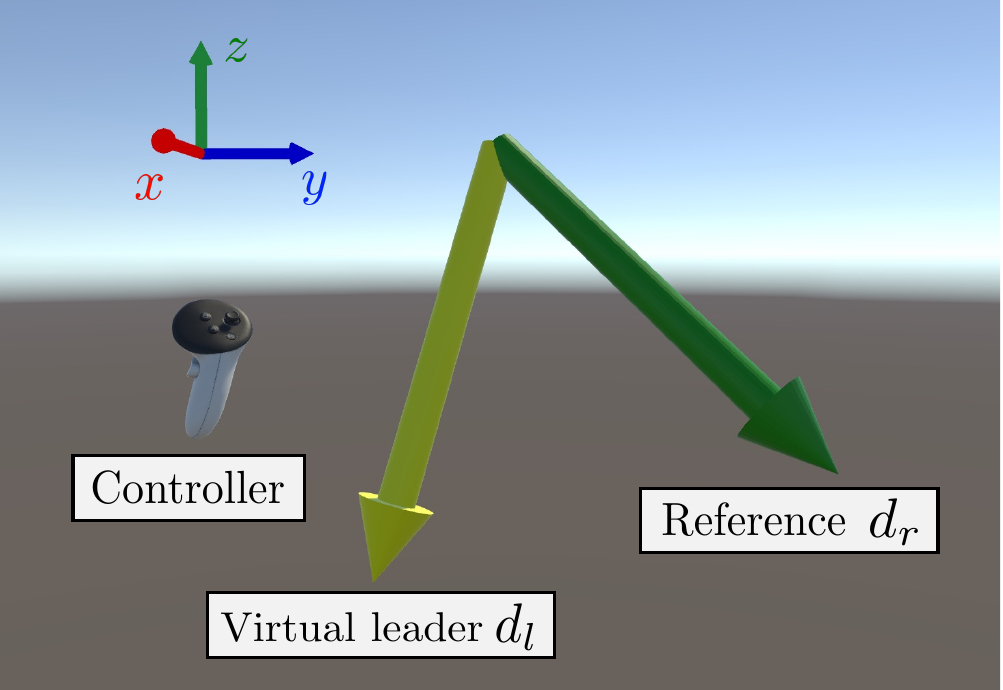}
  \caption{Simulation environment in a virtual space.}
  \label{fig:vr_system}
\end{figure}

In the present simulation, 
the operator tries to regulate the $z$-axis of the frame, 
namely $d_l$, to the externally provided reference $d_r$.
Accordingly, we had the operator visually feedback $d_l$ as shown in Fig.~\ref{fig:vr_system}. 
We initially built a model with three dimensional output $\omega_{\mathrm{h}}^b$. 
As a result, the gain of the third element of $\omega_{\mathrm{h}}^b$ corresponding to the rotation around the $z$-axis
was remarkably smaller than the other two, although we cannot show the results due to the page constraint. 
We thus artificially reset $(\omega_{\mathrm{h}}^b)_3$ to zero in order to simplify the modeling and analysis. 
Note that this constraint does not affect the passivity-based stability analysis presented in Section \ref{sec:passivity-based_control}, 
as the general properties of the attitude dynamics on $SO(3)$ remain preserved.

We had a trial subject conduct the human-in-the-loop simulation after
sufficient trainings. 
The reference $d_r$ was changed randomly at every 15 s. This alignment task was repeated 12 times totaling
180 s; the first 6 trials used for identification and the latter 6 trials for validation.
Since the reference $d_{r}$
is externally given to the operator, the initial response period includes cognitive dead time. 
This phenomenon is specific to the present identification experiments and is unlikely to occur in real operations in which the operator generates the reference internally.
Therefore, we eliminated the data 
from the moment the reference value changed until the operator pressed the start button in the same way as (\cite{Hatanaka:2024}).

\begin{figure}[tb]
  \centering
  \includegraphics[width=1\linewidth]{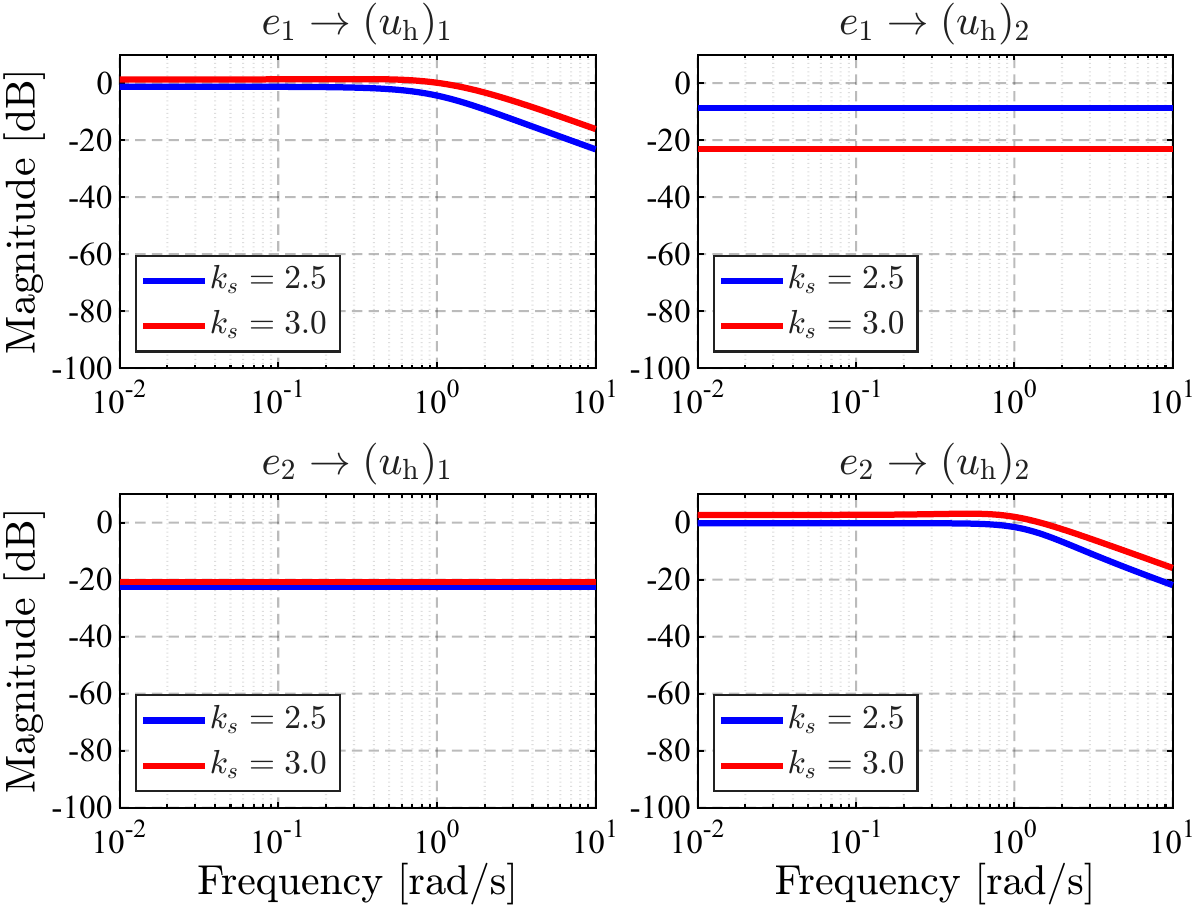}
  \caption{Bode diagrams of the identified human operator models with 2-by-2
  transfer function matrix.}
  \label{fig:bode_diagram_2x2}
\end{figure}

\begin{figure}[tb]
  \centering
  \includegraphics[width=1\linewidth]{
    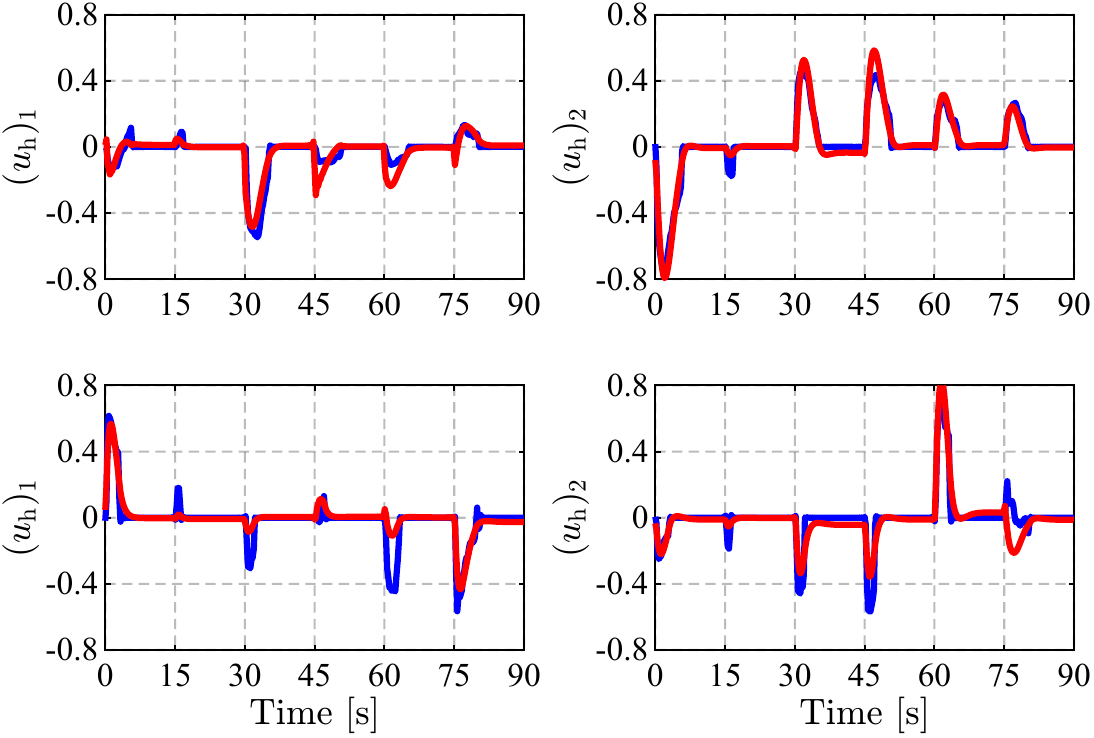
  }
  \caption{Time series data of the verification data (blue) and the model
  outputs (red) on $u_{\mathrm{h}}$ (top: $k_s=2.5$, bottom: $k_s=3.0$).}
  \label{fig:time_series_model_outputs}
\end{figure}

\subsection{Human Modeling}


Let us first identify the continuous-time transfer function that represents the human
behavior. 
The output of the model is set to the first two elements of the body velocity command as
$u_{\mathrm{h}}:=[(\omega_{\mathrm{h}}^b)_{1}\ (\omega_{\mathrm{h}}^b)_{2}]^{\top}$ as the third element is artificially set to zero.
The input to the model is defined as 
 $e:=[\{\text{sk}(R_{l}^{\top} R_{r})^{\vee}\}_{1}\ \{
\text{sk}(R_{l}^{\top} R_{r})^{\vee}\}_{2}]^{\top}$
while excluding the third element of the error, since the operator does not care about the regulation corresponding to this element.
Actually, we built a model with three inputs and two outputs, and confirmed that the contributions of the third element of the error to $u_\mathrm{h}$ are limited. 

The identification is performed utilizing 
MATLAB System Identification Toolbox (MathWorks Inc.).
Prior to identification, downsampling was performed to decimate the sampled data
to $10$ rad/s.
As a result, we obtained a model $u_{\mathrm{h}}(s)=H(s)e(s)$
with a 2-by-2 transfer function matrix $H(s)$.
In the system identification, we used a model with 2 poles and 1 zero for diagonal elements 
and constant non-diagonal elements, which is determined through the same process as (\cite{Hatanaka:2024}).
The bode diagrams of the model with two gains $k_s = 2.5$ and $k_s = 3.0$ are illustrated in Fig.~\ref{fig:bode_diagram_2x2}.
We see from the figure that the diagonal elements are dominant in determining $u_{\mathrm{h}}$, 
which is reasonable from physical correspondence of the input-output variables.



Fig.~\ref{fig:time_series_model_outputs} shows the time response of the model outputs
and validation data.
For $k_s=2.5$, the fit ratios of the model output were 69.40\% for the identification and 
66.15\% for the validation data.
In contrast, for $k_s=3.0$, the ratios were 77.79\% and 51.66\%, respectively.
Although the fit ratio is moderate, the identified model sufficiently captures the dominant input–output behavior, as seen in Fig.~\ref{fig:time_series_model_outputs}, relevant to passivity and stability analysis.

\subsection{Passivity Analysis}
Let us analyze passivity of the identified human models. To verify the
passivity condition (\ref{eq:passivity_human}) presented in Assumption \ref{assum:passivity_human},
we define the passivity index
\begin{align*}
  \nu(\omega) = \frac{1}{2}\lambda_{\min}(H(j\omega)+H^{\mathsf{H}}(j\omega)), \label{eq:passivity_index}
\end{align*}
where the notation ${\mathsf{H}}$ denotes the conjugate transpose and
$\omega$ is the angular frequency. This passivity index facilitates a frequency-wise
passivity analysis. Specifically, it is well known that a linear time-invariant system $H$ is passive if and only if
$\nu(\omega) \geq 0$ holds $\forall\omega$.

The passivity indices $\nu(\omega)$ for the identified human models are illustrated in Fig.~\ref{fig:passivity_index_2x2}.
We see from this figure that the condition $\nu(\omega) \geq 0$ holds for $k_s=3.0$, while it is not satisfied for $k_s=2.5$ at frequencies above approximately 2 rad/s.
This implies that passivity may not be satisfied depending on the parameter settings even for the same operator.
The results suggest that the parameter $k_s$ should be carefully tuned to ensure the human passivity. In addition to the thorough analysis on the relation between $k_s$ and human passivity, extending the present framework to the one accepting possibly passivity-short human operators as done in (\cite{Hatanaka:2024}) is also left as future work.


%% file: text/conclusion.tex
\section{Conclusion}
\label{chp:conclusion}
This paper proposed a novel semi-autonomous rotational motion navigation framework designed based on passivity. 
Closed-loop stability was rigorously established under the assumption of human passivity.
In addition, human-in-the-loop simulations were conducted to identify models of the human operator. Passivity analysis of the identified models revealed that human passivity is sensitive to parameter selection, highlighting the necessity of extending the present framework toward control design that explicitly accounts for passivity-short human operators.



\begin{figure}[tb]
    \centering
    \includegraphics[width=1\linewidth]{
      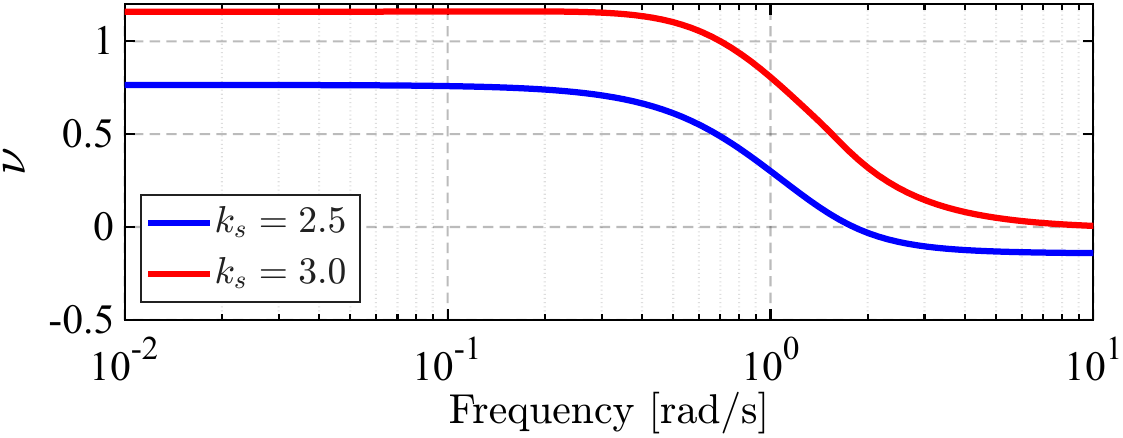
    }
    \caption{Passivity indices of the identified human operator
    model with 2-by-2 transfer function matrix.}
    \label{fig:passivity_index_2x2}
  \end{figure}